\documentclass[12]{article}

\usepackage{amssymb,latexsym,amsmath}
\usepackage{amsmath}     
\usepackage{dsfont}
\usepackage{upgreek}
\usepackage{graphics}
\usepackage{graphicx}
\usepackage{lscape}
\usepackage{multirow}
\usepackage{amssymb}
\usepackage{dsfont}
\usepackage{hyperref}
\usepackage{authblk}
\usepackage[switch,columnwise]{lineno}
\usepackage{graphicx}
\usepackage{caption}
\usepackage{subcaption}
\usepackage{booktabs}
\newtheorem{definition}{Definition}
\newtheorem{theorem}{Theorem}
\newtheorem{proposition}{Proposition}

\begin{document}
\newenvironment{proof}{{\bf Proof:\ \ }}{\qed}
\newcommand{\qed}{\rule{0.5em}{1.5ex}}
\title{A New Count Model Generated from Mixed Poisson Transmuted Exponential Family with an application to Health Care Data}
\author[1]{Deepesh Bhati}
\author[1]{Pooja Kumawat}
\author[2]{E. G\'omez--D\'eniz}

\affil[1] {\textit{Department of Statistics, Central University of
Rajasthan}} \affil[2] {\textit{Department of Quantitative Methods
in Economics and TiDES Institute. University of Las Palmas de Gran
Canaria, Spain.}} \maketitle

\begin{abstract}
In this paper, a new mixed Poisson distribution is introduced. This new distribution is obtained by utilizing mixing process, with Poisson distribution as mixed distribution and Transmuted Exponential distribution as mixing distribution. Some distributional properties like unimodality, moments, over-dispersion, Taylor series expansion of proposed model are studied. Estimation of the parameters using method of moments, method of moments and proportion and maximum likelihood estimation along with data fitting experiment to show its advantage over some existing distribution.
Further, an actuarial applications in context of aggregate claim distribution is discussed. Finally, we  discuss a count regression model based on proposed distribution and its usefulness over some well established model.
\end{abstract}


\vfill
\subsubsection*{Acknowledgements}
EGD was partially funded by grant ECO2013--47092 (Ministerio de
Econom\'ia y Competitividad, Spain).

\section{Introduction}
Recently, count data have drawn attention of many researchers working in different area of Insurance, economics, social sciences and biometrics. For this purpose, traditional models like Poisson, Negative binomial, Geometric and their generalizations were used. But often it has been found that count data exhibits over-dispersion (variance $>$ mean) and long tail behaviour. Hence there is further demand to modify/generalize these traditional models in encounters such problems. In last two decades, many attempts have been made by the researchers to develop new models, one such method which has been widely used to model count data is mixture of the distribution which have been widely used for modelling observed situations whose various characteristics as reflected by the data differ from those that would be anticipated under the simple component distribution, Karlis
and Xekalaki (2005). \\
Many count-data models proposed by mixing Poisson parameter $(\lambda)$ with various continuous distribution which overcomes the problems of under or equi-dispersion, for development in literature of Mixed Poisson distribution see the references below \\
\begin{itemize}
\item Poisson-Gamma (Negative Binomial) -Greenwood and Yule
    (1920)
\item Poisson Beta with specific parameter values (Yule)
    -Simon (1955)
\item Poisson Beta Type-2 -Gurland (1958)
\item Poisson-Exponential Beta -Pielou (1962)
\item Poisson Truncated Poisson -Patil (1964)
\item Poisson Beta Type 1 -Holla and Bhattacharya (1965)
\item Poisson Truncated Gamma -Bhattacharya(1966)
\item Poisson Linear Exponential -Sankaran(1969)
\item Poisson Lindley -Sankaran (1970)
\item Poisson Power Function  -Rai (1971)
\item Poisson Lognormal  -Bulmer (1974)
\item Poisson Generalized Inverse Gaussian -Sichel (1974)
\item Poisson Inverse Gaussian -Sichel(1975)
\item Poisson Gamma Product Ratio(Generalized Waring) -Irwin
    (1975)
\item Poisson Generalized Pareto -Kempton (1975)
\item Poisson-Poisson Distribution(Neyman) -Douglas (1980)
\item Poisson Pearson's Family of Distribution  -Albrecht
    (1982)
\item Poisson Generalized Gamma  -Albrecht (1984)
\item Poisson Truncated Beta Type 2  -Willmot (1986)
\item Poisson Log-Student -Gover and O' Muircheartigh (1987)
\item Poisson Shifted Gamma   -Ruhonen (1988)
\item Poisson Exponential(Geometric) -Johnson et al. (1992)
\item Poisson-Other Discrete Distribution(Neyman) -Johnson et
    al. (1992)
\item Poisson Linear Exponential -Kling and Goovaerts (1993)
\item Poisson Inverse Gamma - Willmot (1993)
\item Poisson Truncated Gamma -Willmot (1993)
\item Poisson Pareto -Willmot (1993)
\item Poisson Shifted Pareto -Willmot (1993)
\item Poisson Modified Bessel -Ong and Muthaloo (1995)
\item Poisson-Power Variance -Hougaard et al. (1997)
\item Poisson-Lomax -Al-Awadhi and Ghitany (2001)
\item Poisson Lindley Distribution(Size Biased) -Al-Multairi
    (2008)
\item Zero Truncated Poisson-Lindley -Ghitany et al. (2008)
\item Poisson Generalized-Lindley -Mahmoudi and Zakerzadah
    (2010)
\item Poisson-Lindley-Beta -G\'{o}mez--D\'eniz et al. (2014)
\item Poisson-Marshall-Olkin-Generalized Exponential
    -G\'{o}mez--D\'eniz et al. (2015)
\end{itemize}

The above monographs presents more flexible distributions as they can be used as building blocks for improving count data models. In the present paper, a new discrete distribution is obtained by mixing Poisson distribution with Transmuted-Exponential distribution proposed by Shaw and Buckley (2007). \\
\noindent To the best of our knowledge, we have not come across any literature on discrete distribution, particularly mixed Poisson distribution, where transmuted family is being used, except Charkraborty and Bhati(2015) who proposed the discrete version of Transmuted exponential family. \\

The rest of the paper is structured as follows. Section 2 describes the theoretical development of the new count distribution, including some properties, different methods of estimation are shown in Section 4. An actuarial application of the proposed model is examined in Section 5. Finally in Section 6, application of the proposed model in count data analysis and count regression model is presented. Some comments and conclusions are drawn in Section 7.

\section{Proposed model}
Shaw and Buckley (2007) proposed a novel technique to introduce skewness and kurtosis into symmetric as well as to other distribution. In this techniques they use ``transmutation map'', which is functional composition of cumulative distribution of one distribution with the quantile function of another. One such member of this transmuted family is ``Transmuted Exponential Distribution'' $\left(\mathcal{TED}(\alpha,\theta)\right)$ whose density function given as
\begin{equation} \label{1}
f(x)=\begin{cases}
0  \qquad & \text{for} \qquad  x<0, \\
\bar{\alpha}\theta e^{-\theta x}+2\alpha \theta e^{-2 \theta x} & \text{for}  \qquad  x \ge 0,
\end{cases}
\end{equation}
with $|\alpha|<1$ and $\theta>0.$ \\
\noindent Considering the fact that $\left(\mathcal{TED}(\alpha,\theta)\right)$ possess wide range of statistical properties as compared to exponential distribution, we introduce a new mixed Poisson distribution assuming $\left(\mathcal{TED}(\alpha,\theta)\right)$ as prior distribution for Poisson parameter $(\lambda)$. The definition of purposed model is as follows\\

\begin{definition} A random variable $X$ is said to follow Poisson-Transmuted
Exponential distribution if it possess following stochastic
representation
\begin{equation*}
\begin{aligned}
X|\lambda &\sim & Po(\lambda) \\
\lambda|\alpha,\theta & \sim & \mathcal{TED}(\alpha,\theta)
\end{aligned}
\end{equation*}
for $\lambda >0$, $|\alpha| \le 1$, $\bar{\alpha}=1-\alpha$ and $\theta>0$. We denote unconditional distribution of $X$ by $\mathcal{PTED}(\alpha,\theta)$ and its pmf is given theorem 1. \\
\end{definition}

\begin{theorem} If $X\sim \mathcal{PTED}(\alpha,\theta)$, then probability mass function (pmf) of \textit{X} is \\
\begin{equation*}
P(X=x;\alpha,\theta)=\theta \left(\frac{\bar{\alpha}}{(1+\theta)^{x+1}}+\frac{2 \alpha}{(1+2\theta)^{x+1}} \right),\qquad  x=0,1,\dots,
\end{equation*}
with $|\alpha| \le 1$ and $\theta >0$.
\end{theorem}

\noindent \textit{Proof:} If $X|\lambda \sim Pois(\lambda)$ and $\lambda \sim \mathcal{TED}(\alpha,\theta)$ then
\begin{align} \label{2}
\nonumber P\left(X=x\right)=\; &\theta\int\limits_{0}^{\infty}P(X=x|\lambda)f_{\Lambda}(\lambda;\alpha,\theta) d\lambda \\
\nonumber  = \;& \theta\int\limits_{0}^{\infty}\frac{e^{-\lambda}\lambda^{x} }{x!}\theta e^{-\theta\lambda}\left(\bar{\alpha}+2\alpha e^{-\theta}\lambda\right)d\lambda \\
\nonumber =\;&\frac{\theta}{x!}\int\limits_{0}^{\infty}e^{-\lambda}\lambda^{x}\left(\bar{\alpha}+2\alpha e^{-\theta}\lambda\right) e^{-\theta\lambda} d\lambda\\
=\; &\theta\left(\frac{\bar{\alpha}}{(1+\theta)^{x+1}}+\frac{2\alpha}{(1+2\theta)^{x+1}}\right),
\end{align}
where $|\alpha| \leq 1$ and $\theta > 0$.
\begin{flushright}
$\blacksquare$
\end{flushright}

\noindent \textbf{Remarks}
\begin{enumerate}
\item[i] It is easy to see that the p.m.f. in (\ref{2}) can be rewritten as a finite mixture of two
geometric distribution in the form,
\begin{eqnarray} \label{3}
p_x=\bar\alpha\; \mbox{Geom}\left(\frac{1}{1+\theta}\right)
+\alpha\;\mbox{Geom}\left(\frac{1}{1+2\theta}\right).\label{mixture}
\end{eqnarray}
\item[ii] For $\alpha\rightarrow 0$ pmf (\ref{2}) reduces  to Geometric distribution $Geo\left(\frac{1}{1+\theta}\right) $.
\item[iii] For $\alpha\rightarrow 1$ pmf (\ref{2}) reduces to Geometric distribution $\left(\frac{1}{1+2\theta}\right)$.
\end{enumerate}
\noindent Further, the cumulative distribution function and the survival function of  $X\sim$ $\mathcal{PTED}(\alpha,\theta)$ can be given as
\begin{align} \label{4}
\nonumber F_{X}\left( x\right) =P(X\leq x)= &\sum\limits_{n=0}^{x}\theta\left( \frac{\bar{\alpha}}{(1+\theta)^{n+1}}+\frac{2\alpha}{(1+2\theta) ^{n+1}}\right)\\
= & 1-\left(\frac{\bar{\alpha}}{(1+\theta)^{x+1}}+\frac{\alpha}{(1+2\theta)^{x+1}}\right)
\end{align}
\noindent and
\begin{equation} \label{5}
S_{X}(x)=P(X \ge x)=\frac{\bar{\alpha}}{(1+\theta)^{x}}+\frac{\alpha}{(1+2\theta)^{x}}
\end{equation}\\
\noindent Figure 1 shows the pmf of $\mathcal{PTED}(\alpha,\theta)$ for different values of $\alpha, \theta$ which confirms the unimodality of random variable, moreover, this can proved mathematically in proposition 1. \\

\begin{figure} \label{f1}
\centering
\includegraphics[width=1\textwidth]{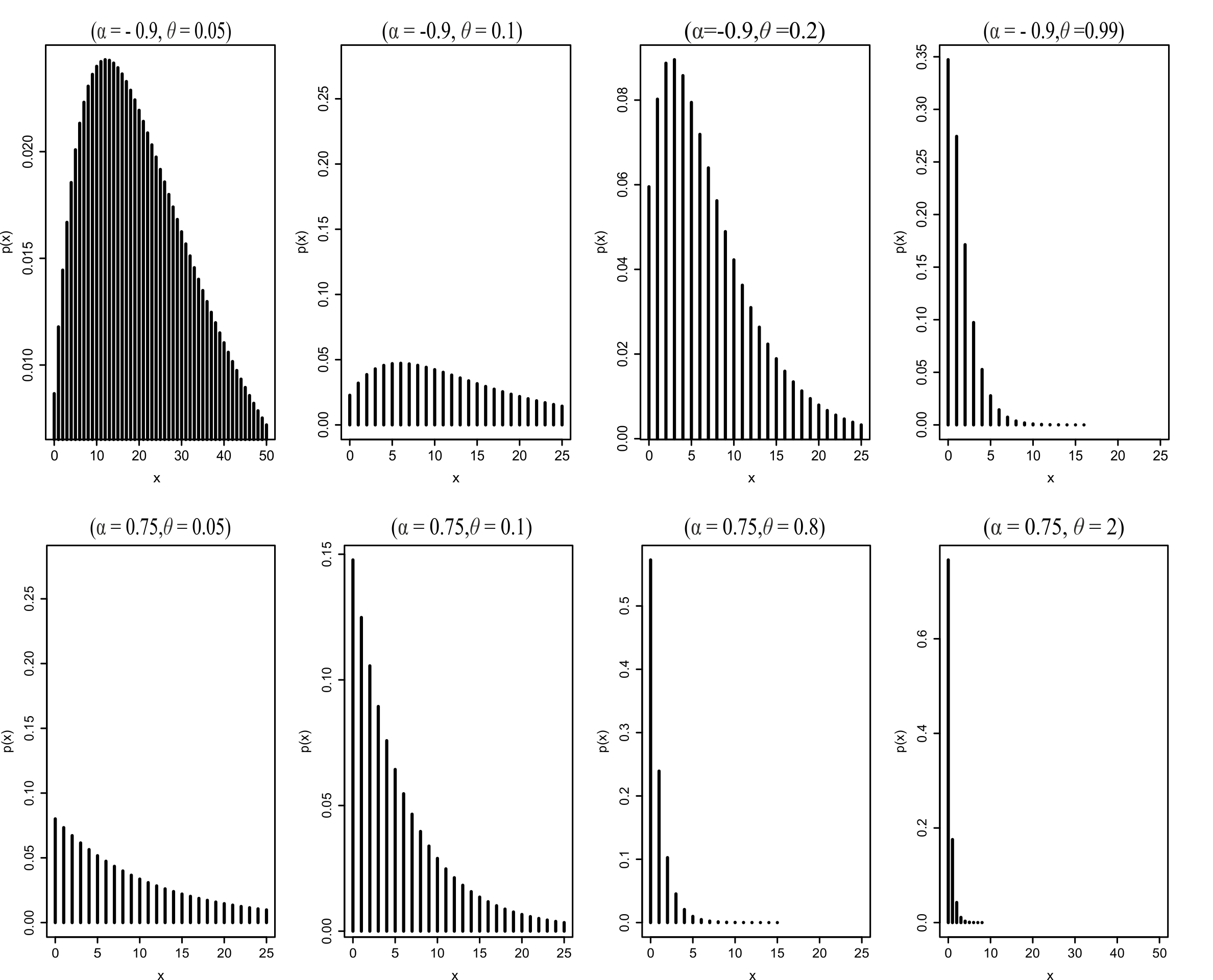}
\caption{PMF plot for different values of $\alpha$ and $\theta$.}
\end{figure}

\begin{proposition} The Proposed distribution with p.m.f. (\ref{2}) is unimodal.
\end{proposition}
\noindent \textit{Proof:} The probabilities defined in (\ref{2}) can also be computed by following recursive relation
\begin{equation} \label{6}
p\left(x+1\right)=\frac{1}{1+\theta}\left(\frac{\bar{\alpha}+2\alpha\left( \frac{1+\theta}{1+2\theta}\right)^{x+1} }{\bar{\alpha}+2\alpha\left( \frac{1+\theta}{1+2\theta}\right)^{x+2}}\right) p\left(x\right)
\end{equation}
with\indent
 \begin{equation*}
 p(0)=\frac{\theta\left(1+\alpha+2 \theta\right)}{\left(1+\theta\right)\left(1+2\theta\right)}.
 \end{equation*}\\
\noindent Moreover, it can be observe that (\ref{6}) possess relation
\begin{eqnarray} \label{7}
\begin{cases} p_{x+1} \ge p_x \qquad & \forall \quad x\le x^*, \\
p_{x+1} \le p_x \qquad & \forall \quad x\ge x^*,
\end{cases}
\end{eqnarray}
for $x^*$(mode of $\mathcal{PTED}$ r.v.) depending on parameters $\alpha$ and $\theta$, which indicate the unimodality, see Keilson and Gerber (1971).\\

Further,
\begin{equation*}
p^2_{x}-p_{x+1} p_{x-1} = -\frac{2\alpha(\bar{\alpha})\theta^2}{(1+\theta)^{x+2}(1+2\theta)^{x+2}}.
\end{equation*}

For $-1<\alpha<0$, the distribution is infinitely divisible since
$p_x^2-p_{x-1}p_{x+1}<0$ and $p_0\neq 0,\;p_1\neq 0$, see Warde
and Katti (1971) for details. Then, in this case the distribution
has its mode at zero. The fact that
$\left\{p_x/p_{x-1}\right\},\;x=1,2,\dots$, forms a monotone
increasing sequence in this case requires that $\{p_x\}$ be a
decreasing sequence (see Johnson and Kotz, 1982, p.75), which is
congruent with the zero vertex of the new distribution for
$-1<\alpha<0$. Moreover, as any infinitely divisible distribution
defined on non-negative integers is a compound Poisson
distribution (see Proposition 9 in Karlis and Xekalaki, 2005), we
conclude that the new pmf presented in this paper is a compound
Poisson distribution.

Furthermore, the infinitely divisible distribution plays an important role in many areas of statistics, for example, in
stochastic processes and in actuarial statistics. When a distribution $G$ is infinitely divisible then for any integer
$x\geq 2$, there exists a distribution $G_x$ such that $G$ is the $x$--fold convolution of $G_x$, namely, $G = G_x^{\ast x}$.
Also, when a distribution is infinitely divisible an upper bound for the variance can be obtained when $-1<\alpha<0$ (see Johnson and Kotz, 1982, p.75), which is given by
\begin{eqnarray*}
var(X)\leq \frac{p_1}{p_0}.
\end{eqnarray*}

\noindent Some implications of infinite divisibility are as follows:
\begin{enumerate}
\item For $k,\;m=1,2,\dots$,
\begin{eqnarray*}
{k+m\choose m}p_{k+m}p_0\geq p_k p_m.
\end{eqnarray*}

See Steutel and van Harn (2004), p. 51, Proposition 8.4.

\item For all $k=0,1,\dots$, $p_x\leq \exp(-1)$. See Steutel and van Harn (2004), p. 56, Proposition 9.2.

\item The cumulants of an infintely divisible distribution on the set of non--negative integers (as far as they exist) are non negative, see Steutel and van Harn (2004), p. 47, Corollary 7.2. This will imply that the skewness of the new distribution is positive, since the third cumulant equals the third central moment.

\item The distribution is strictly log-concave and strongly unimodal, see theorem 3 in Keilson and Gerber (1971).
\end{enumerate}

As indicate in above remark, since (\ref{2}) is log-concave and zero vertex for some value of parameter $\alpha$ and $\theta$, we can determine the mode$(x^*)$ of $\mathcal{PTED}(\alpha,\theta)$, which will be determine such the $P(X=x)$ is increasing on $(0,1,\cdots,x^*)$ and decreasing on $(x^*+1,\cdots)$. It can be easily verified-
\begin{enumerate}
\item For $-1<\alpha<-\frac{(1+2\theta)^2}{3+4\theta}$, the pmf (\ref{2}) is unimodal and the mode is at $x_0=0$.
\item For $ -\frac{(1+2\theta)^2}{3+4\theta}<\alpha<1$, the mode of pmf (\ref{2}) will be at $x^*+1$, with
\indent $x^*=\left[\log_{\theta^*}\left(\frac{\alpha-1}{4 \alpha}\right)-2\right]$, where $\theta^*=\frac{1+\theta}{1+2\theta}$. Here [.] denotes the integer part.
\item Further if $\left[\log_{\theta^*}\left(\frac{\alpha-1}{4 \alpha}\right)-2 \right]$ is an integer, then pmf (\ref{2}) will be bimodal and the mode is at $x^*$ and $x^*+1$ respectively.
\begin{flushright}
$\blacksquare$
\end{flushright}
\end{enumerate}

\noindent Moreover, the new distribution is as a mixed Poisson distribution it has a heavier tail than a Poisson distribution with the same mean. Now, let $\Pr(X=x)$ be the probability function (\ref{2}) and $\Pr(X|m)$ be the probability function of a simple Poisson distribution with the same mean, say $m$. Then, as shown by Feller (1943), $\Pr(X=0)>\Pr(0|m)$ and $\Pr(X=1)/\Pr(X=0)\leq \Pr(1|m)/\Pr(0|m)=m$. The asymptotic tail behavior of Poisson
distributions was studied for Willmot (1990). \\

\noindent In the following results, we present various properties such as Taylor Expansion for $\mathcal{PTE}$ probabilities , probability generating function, moments, hazard function of $\mathcal{PTE}(\alpha,\theta)$ \\

\subsubsection*{Taylor expansion of $\mathcal{PTE}(\alpha,\theta)$ probabilities}
\noindent Ong(1995) gave the Taylor series expansion of mixed poisson distribution which is stated as: \\
\noindent Let $g(x)$ be the probability density function (pdf) of the mixing distribution of a mixed Poisson distribution. If g(x) has a finite $n^{th}$ derivative at the point $k \ge 1$ for all $n$, then the mixed Poisson pmf $P(k)$ has the formal expansion
\begin{equation} \label{8}
P(k)=g(k)+\frac{1}{k}\left(\mu_2 \frac{f^{(2)}(k)}{2!}+\mu_3 \frac{f^{(3)}(k)}{3!}+ \cdots +\mu_n \frac{f^{(n)}(k)}{n!} \right)+\cdots
\end{equation}
where $f(x)=x g(x)$, $f^{(i)}(x)=\frac{d^if(x)}{dx^i}$ and $\mu_i$ is the $i^{th}$ moment about the mean of the gamma random variable $X$ with scale and shape parameters 1 and $k$ respectively. \\
In the proposed model, considering $f(x)=x g(x)=x\theta e^{-\theta x}\left(\bar{\alpha}+\alpha e^{-\theta x} \right)$, and $f^{(i)}(x)=\theta^i e^{-2 \theta x}\left((-1)^{i-1}e^{\theta x}(\bar{\alpha}(i-\theta x)+ (-1)^{i-1} 2^{i-1}\alpha(i-2\theta x)\right)$ and using (7), the probability of $\mathcal{PTED}(\alpha,\theta)$ can easily be obtained. \\

\begin{proposition} The probability generating function(pgf) of random variable
$X$ defined in (\ref{2}) is given as
\begin{equation} \label{9}
P_X(t)= \frac{\theta(1-t)(1+\alpha)+2\theta^2}{(1-t+\theta)(1-t+2\theta)}.
\end{equation}
\end{proposition}
\begin{proof}
The proof is  straight forward after using the definition
$P_X(t)=E(t^X)$ and (\ref{2}).
\end{proof}

\begin{proposition} The $r^{th}$ raw moment of $\mathcal{PTED}(\alpha,\theta)$
is given by
\begin{equation} \label{10}
E(X^{r})= \frac{\theta(\bar{\alpha})\Phi(\frac{1}{1+\theta},-r,0)}{1+\theta}+\frac{2\theta\alpha \Phi(\frac{1}{1+2\theta},-r,0)}{1+2\theta},
\end{equation}
where $\Phi(z,s,a)=\sum\limits_{k=0}^{\infty} z^k(k+a)^{-s}$ is
Hurwitz-Lerch transcendent function.
\end{proposition}

In particular, the first four raw moments of $X$ can be obtained easily by putting r=1,2,3,4 in (\ref{10}) and are as follows\\
\begin{equation*}
\begin{split}
\mu^{'}_{1}=&\frac{2-\alpha}{2\theta},\\
\mu^{'}_{2}=&\frac{4-3 \alpha +2 \theta -\alpha  \theta }{2 \theta ^2},\\
\mu^{'}_{3}=&\frac{24+24\theta+4\theta^{2}-21\alpha-18\theta\alpha-2\theta^{2}\alpha}{4\theta^{3}},\\
\mu^{'}_{4}=&\frac{48+72\theta+28\theta^{2}+2\theta^{3}-45\alpha-63\theta\alpha-21\theta^{2}-\theta^{3}\alpha}{2\theta^{4}},\\
\end{split}
\end{equation*}
whereas other measures like variance$(\mu_2)$, coefficient of skewness $\left( \gamma_1=\frac{\mu_{3}}{\mu^{3/2}_{2}}\right) $ and kurtosis  $\left(\gamma_2= \frac{\mu_{4}}{\mu^{2}_{2}}\right) $ are as follows
\begin{align}
\mu_{2}=&\frac{4+2\theta(2-\alpha)-\alpha(2+\alpha)}{4\theta^{2}},\\
\gamma_1 =& \frac{2(8+4\theta(3+\theta)-3\alpha-2\theta\alpha(3+\theta)-3\alpha^{2}(1+\theta)-\alpha^{3})}{\left[4-2\theta(\alpha-2)-\alpha(\alpha+2) \right]^{3/2} },\\
\gamma_2=&\frac{16(1+\theta)(9+\theta(9+\theta))-8(1+\theta)(9+\theta(12+\theta))\alpha-8(6+\theta(9+2\theta))\alpha^{2}-12(1+\theta)\alpha^{3}-3\alpha^{4}}{(4-2\theta(\alpha-2)-\alpha(\alpha+2))^{2}}.
\end{align}
In figure 2 the contour plot of Mean, Variance, Skewness and Kurtosis for parameters $\alpha$ and $\theta$ were shown.

\begin{figure} \label{f2}
\centering
\includegraphics[width=1\textwidth]{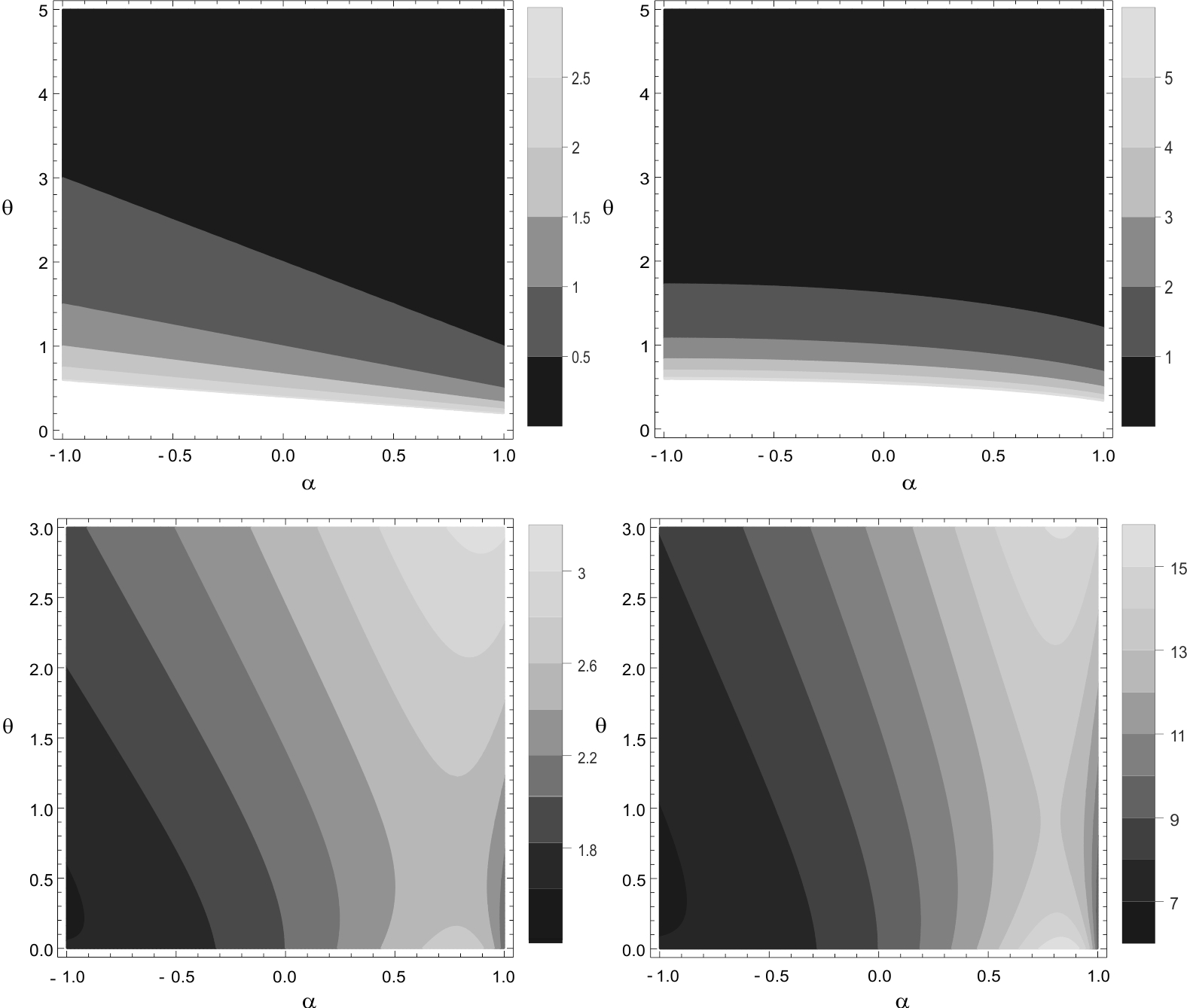}
\caption{From left top-in clockwise direction we present contour plot of Mean, Variance, Kurtosis and Skewness of $\mathcal{PTED}(\alpha,\theta)$}
\end{figure}

Further, the coefficient of variation(C.V.) of the distribution
comes out to be
\begin{equation*}
C.V. = \frac{\sigma}{\mu} =\frac{\sqrt{4-2\theta(\alpha-2)-\alpha(\alpha+2)}}{2-\alpha}.
\end{equation*}

\section{Steps to generate $\mathcal{PTED}(\alpha,\theta)$ random variate}

As the proposed model is derived from mixed Poisson distribution, following algorithm can be used to get $\mathcal{PTED}(\alpha, \theta)$ random variable \\

\subsubsection*{Algorithm}
\begin{enumerate}
\item[Step 1] Generate $u_i$ from $U(0,1)$.
\item[Step 2] Corresponding to each $u_i$, determine
    $\lambda_i =
    \frac{1}{\theta}\log\left(\frac{1-\alpha+\sqrt{1+2\alpha-4u_i\alpha+\alpha^2}}{2(1-u_i)}\right)$.
\item[Step 3] Hence, generate $x_i$ from $Poisson(\lambda_i)$.
\end{enumerate}

\section{Estimation}
In this section we discuss three methods to estimate the parameters $\alpha$ and $\theta$, In the first subsection the method of moment is presented. further in second section, method of proportion and moment is discussed and finally Maximum Likelihood method is shown as last subsection.

\subsection{Method of moments}
Given a random sample $x_1,x_2, \cdots, x_n$ of size $ n $ from (1), the moment estimates, $\tilde{\alpha}$ and $\tilde{\theta}$, of $\alpha$ and $ \theta $ can be obtained by solving the following equations
\begin{equation}
m_1=\mu _1'=\frac{2-\alpha}{2 \theta}
\qquad \text{and} \qquad
m_2=\mu _2'=\frac{4-3\alpha+2\theta-\alpha\theta}{2\theta ^2}.
\end{equation}
where $ m_1 $ and $ m_2 $ are the first and second sample moments. Solving the above equations, we get
\begin{align}
\nonumber \tilde{\alpha}&=2+\frac{3 m_1^2}{m_1-m_2}+\frac{m_1 \sqrt{4 m_1+9 m_1^2-4 m_2}}{m_1-m_2},\\
\tilde{\theta}&=\frac{-3 m_1-\sqrt{4 m_1+9 m_1^2-4 m_2}}{2 (m_1-m_2)}.
\end{align}

\begin{theorem} For fixed $\alpha$, the estimator $\tilde{\theta}$ of
$\theta$ is positively biased, i.e. $E(\tilde{\theta})>\theta$.
\end{theorem}
\begin{proof} Let $\tilde{\theta}=g(\bar{X})$ and
$g(t)=\frac{2-\alpha}{2t}$ for $t>0$. Then
\begin{equation*}
g''(t)=\frac{2-\alpha}{t^3} >0.
\end{equation*}

Therefore, $g(t)$ is strictly convex. Thus, by Jensens inequality,
we have $E\left(g(\bar{X})\right)>g\left( E(\bar{X})\right)$.
Finally, since $g\left(
E(\bar{X})\right)=g\left(\frac{2-\alpha}{2\theta}\right)=\theta$,
hence we obtain $E(\tilde{\theta})>\theta$.
\end{proof}

\subsection{Method of proportion and moment}
In this method, we compare the sample proportion of zero($p_0$) and sample mean ($\bar{x}$) with the population proportion of zero and population mean, i.e, estimated $\breve{\alpha}$ and $\breve{\theta}$ will be obtained by solving the following two equations
\begin{equation}
p_0=\theta \left(\frac{\bar{\alpha}}{1+\theta}+\frac{2\alpha}{1+2\theta} \right) \qquad \text{and} \qquad  \bar{x}=\frac{2-\alpha}{2\theta}.
\end{equation}

Thus, solving the above two equations, we obtain the estimates $\breve{\alpha}$ and $\breve{\theta}$ as follows
\begin{align}
\nonumber \breve{\alpha}=& \frac{1}{2} \left(4-\frac{3 \bar{x}}{\bar{x}+p_0-1}+\frac{3 \bar{x} p_0}{\bar{x}+p_0-1}-\frac{\bar{x} \sqrt{9-10 p_0-8 \bar{x} p_0+p_0^2}}{\bar{x}+p_0-1}\right), \\
\breve{\theta}=& \frac{3-3 p_0+\sqrt{9-10 p_0-8 \bar{x} p+p_0^2}}{4 (\bar{x}+p_0-1)}.
\end{align}

\subsection{Maximum likelihood estimation}
Let $x_{1},x_{2},\cdots,x_{m}$ be an random observation of size m from our proposed Poisson-Transmuted exponential distribution. The log-likelihood function for the vector of parameter $\Theta=(\alpha,\theta)^\top$ can be written as \\
\begin{equation*}
l(\alpha,\theta|x)= m\log\theta+ \sum\limits_{i=1}^{m}\log\left(\frac{\bar{\alpha}}{(1+\theta)^{x_{i}+1}}+\frac{2\alpha}{(1+2\theta)^{x_{i}+1}} \right).
\end{equation*}

The normal equations can be obtained by taking the first
derivative with respect to both parameters. These equations are
given by
\begin{align}
\frac{\partial l}{\partial \theta}=&\frac{n}{\theta }-\sum\limits_{x=0}^m \frac{(x+1)\left(\frac{(\bar{\alpha} )}{(1+\theta )^{x+2}}+\frac{4 \alpha  }{(1+2 \theta )^{x+2}} \right)}{\left(\frac{(1-\alpha )}{(1+\theta )^{x+1}} +\frac{2 \alpha }{(1+2 \theta )^{x+1}}\right)}=0, \\
\frac{\partial l}{\partial \alpha}=&\sum\limits_{x=0}^m \frac{\left(-\frac{1}{(1+\theta )^{x+1}}+\frac{2 }{(1+2 \theta )^{x+1}} \right)}{\left(\frac{(\bar{\alpha} )}{(1+\theta )^{x+1}} +\frac{2 \alpha }{(1+2 \theta )^{x+1}}\right)}=0.
\end{align}

The solution of above two equations gives the maximum likelihood
estimator of parameter $\alpha$ and $\theta$ and can be solved
numerically or direct numerical search for global maximum  of the
log likelihood surface. On may also use nlm() function in open
source R-software for determine the global maxima. Further, the
second order derivative of log-likelihood function are

\begin{equation}
\frac{\partial^2 l}{\partial \alpha^2 }= -\sum_{i=1}^n \left(\frac{\left(-\frac{\theta }{(1+\theta )^{x_i+1}} +\frac{2 \theta  }{(1+2 \theta )^{x_i+1}}\right)^2}{\left(\frac{(1-\alpha ) \theta }{(1+\theta )^{x_i+1}}+\frac{2 \alpha  \theta }{(1+2 \theta )^{x_i+1}}\right)^2}\right),
\end{equation}

\begin{align}
\nonumber \frac{\partial^2 l}{\partial \theta \partial \alpha}= \sum _{i=1}^n & \left(\frac{-\frac{1}{(1+\theta )^{x_i+1}}+\frac{2}{(1+2 \theta )^{x_i+1}} +\frac{\theta }{ (1+\theta )^{x_i+2}}(x_i+1)-\frac{4 \theta }{(1+2\theta )^{x_i+2}}(x_i+1)}{\frac{(1-\alpha ) \theta }{(1+\theta )^{x_i+1}} +\frac{2 \alpha \theta  }{(1+2 \theta )^{x_i+1}}} \right. \\
&\left.-\frac{\theta \left(-\frac{1}{(1+\theta )^{x_i+1}}+\frac{2}{ (1+2 \theta )^{x_i+1}}\right) \left(\frac{(\bar{\alpha }) (1-\theta x_i)}{(1+\theta)^{x_i+2}}+\frac{2 \alpha  (1-2 \theta  x_i)}{(1+2 \theta )^{x_i+2} }\right)}{\left(\frac{(\bar{\alpha }) \theta }{(1+\theta )^{x_i+1}} +\frac{2 \alpha \theta}{(1+2 \theta )^{x_i+1}}\right)^2}\right),
\end{align}

\begin{align}
\nonumber \frac{\partial^2 l}{\partial \theta^2}= -\frac{n}{\theta^2}-\sum _{i=1}^n & \left(\frac{\left(\frac{(1-\alpha)}{ (1+\theta )^{x_i+2}} (1-\theta x_i)+\frac{2\alpha}{(1+2\theta )^{x_i+2}} (1-2\theta x_i)\right)^2}{\left(\frac{(1-\alpha ) \theta }{(1+\theta)^{x_i+1}} +\frac{2 \alpha  \theta}{(1+2 \theta )^{x_i+1}}\right)^2}\right.\\
 &\left. +\frac{(x_i+1) \left(\frac{(1-\alpha )}{(1+\theta)^{x_i+3}} (2-\theta x_i)+\frac{8 \alpha  }{(1+2 \theta )^{x_i+3}}(1-\theta x_i)\right)}{\frac{(1-\alpha ) \theta  }{(1+\theta )^{1+x_i}}+\frac{2 \alpha  \theta }{(1+2\theta )^{1+x_i}}}\right).
\end{align}

The fisher information matrix can be computed by using the approximations
\begin{align}
\nonumber E\left(\frac{\partial^2 l}{\partial \alpha^2} \right) \approx \nonumber &\left(\frac{\partial^2 l}{\partial \alpha^2} \right)_{\hat{\alpha},\hat{\theta}},  \\
\nonumber E\left(\frac{\partial^2 l}{\partial \alpha \partial \theta} \right) \approx \nonumber&\left(\frac{\partial^2 l}{\partial \alpha \partial \theta} \right)_{\hat{\alpha},\hat{\theta}},  \\
\nonumber E\left(\frac{\partial^2 l}{\partial \theta^2} \right) \approx \nonumber &\left(\frac{\partial^2 l}{\partial \theta^2} \right)_{\hat{\alpha},\hat{\theta}},
\end{align}
where $\hat{\alpha}$ and $\hat{\theta}$ be the maximum likelihood estimator of $\alpha$ and $\theta$.

\section{Collective risk model}
In non-life insurance portfolio, say, motor insurance, the aggregate loss $(S)$ is a random variable defined as sum of claims incurred in a certain period of time. Let us consider now the following actuarial model. Let $X$ be the number of claims in a portfolio of policies in a time period. Let $Y_i$, $i=1,2,\dots$ be the amount of the $i$--th claim and $S=Y_1+Y_2+\dots+Y_X$ the aggregate claims generated by the portfolio in the period under consideration. As usual, two fundamental assumptions are made in risk theory: (1) The random variables $Y_1,Y_2,\dots$ are independent and identically distributed and follows a discrete (continuous) distribution with pmf (pdf) $h(y)$ and  (2) The random variables $X,Y_1,Y_2,\dots$ are mutually independent. The distribution of the aggregate claims $S$ is called the compound distribution and assuming that $Y_i,\;i=1,2,\dots X$ are discrete random variables, the pdf of $S$ is $f_S(y)=\sum_{x=0}^\infty p_x h^{*n}(y)$, where $h^{\ast x}(\cdot)$ denotes the $x$--fold convolution of $h(\cdot)$ and
$p_x$ is given in (\ref{2}).
There exists an extensive literature dealing with compound mixture Poisson distributions see Willmot, (1986,1993) and Antzoulakos and Chadjiconstantinidis (2004). An extensive review of the topic can be found in Sundt and Vernic (2009). Based on the recursion provided by Panjer (1981) for the Poisson distribution, Sundt and Vernic (2009, chapter 3, p.68) developed a simple algorithm to provide the probabilities of the random variable $S$ when the amount of the single claim follows a discrete distribution with pmf $h(x)$.\\
From the Panjer (1981) recursion for the total claim amount when the pmf of the Poisson distribution  is assumed as the
distribution of the number of claims is given by
\begin{eqnarray}
f_S(x|\lambda)=\frac{\lambda}{x}\sum_{y=1}^{x} y h(y) f_S(x-y|\lambda), \quad x=1,2,\dots\label{recursion}
\end{eqnarray}
while $f_S(0|\lambda)=e^{-\lambda}$. Following Sundt and Vernic (2009, p.68) we get, after multiplying in (\ref{recursion}) in both sides by $\lambda^i f(\lambda)$ and integrating we have that,
\begin{eqnarray*}
f_S^i(y)=\frac{1}{y}\sum_{x=1}^{y}y h(x) f_S^{i+1}(y-x),\quad i=0,1,\dots; y=1,2,\dots
\end{eqnarray*}
where $f_S^{i}(y)=\int_0^{\infty}\lambda^i
f_S(y|\lambda)f(\lambda)d\lambda$. Now, starting with
$f_S(0)=p_0$,  the probabilities $f_S(1), f_S(2),\dots$  can be evaluated by the algorithm described in Sundt and Vernic (2009,p.68) and having into account (\ref{2}).\\

For more detail on classic risk model, see Freifelder (1974), Rolski et al. (1999), Nadarajah and Kotz (2006a and 2006b) and reference therein. Here, we consider two such situations: In first situation, the primary distribution is as defined in Section 2 and claim severity distribution as exponential distribution with parameter ($\alpha$) and as we know, Erlang loss distribution may arise in insurance settings when the individual claim amount is the sum of exponentially distributed claims hence in second situation, Erlang distribution with parameters $r$ and $\alpha$ is considered as secondary distribution.  \\

\begin{theorem} If we assume a Poisson Transmuted Exponential distribution with parameter $(\alpha,\theta)$ as primary distribution and and an Exponential distribution with parameter $(\lambda)$ as secondary distribution, then the pdf of aggregate loss random variable $S= \sum\limits_{i=0}^{X}Y_{i}$ is given by \\
\begin{equation*}
f_{S}(y)=\theta \left[\frac{\bar{\alpha}\lambda e^{\frac{-\theta\lambda y}{1+\theta}} }{\left(1+\theta\right)^{2}}+ \frac{2\alpha\lambda e^{\frac{-2\theta\lambda y}{1+2\theta}} }{\left(1+2\theta\right)^{2}}\right]\\  , \qquad  \text{for} \quad y > 0,
\end{equation*}
whereas,
\begin{equation*}
f_{S}(0)=\theta\left(\frac{\bar{\alpha}}{(1+\theta)}+\frac{2\alpha}{(1+2\theta)} \right).
\end{equation*}
\end{theorem}
\begin{proof} By assuming that the claim severity follows an exponential distribution with parameter  $\lambda >0$, since the $n^{th}$ fold convolution of exponential distribution is gamma distribution with parameter m and $\lambda$, the $n^{th}$fold convolution is given by \\
\begin{equation*}
f^{*x}(y)=\frac{\lambda^{x}}{(x-1)!}y^{x-1}e^{-\lambda y} , x = 1,2,...
\end{equation*}

Then the pdf of random variable $S$ is given by \\
\begin{equation*}
\begin{split}
f_{S}(y)=& \sum\limits_{x=1}^{\infty}\theta\left( \frac{\bar{\alpha}}{(1+\theta)^{x+1}}+\frac{2\alpha}{(1+2\theta)^{x+1}}\right) \frac{\lambda^{x}y^{x-1}e^{-\lambda y}}{(x-1)!}\\
=\;&\theta e^{-\lambda y}\sum\limits_{x=1}^{\infty}\theta\left( \frac{\bar{\alpha}}{(1+\theta)^{x+1}}\lambda^{x}\frac{y^{x-1}}{(x-1)!}+\frac{2\alpha}{(1+2\theta)^{x+1}}\lambda^{x}\frac{y^{x-1}}{(x-1)!}\right)\\
=\;&\theta e^{-\lambda y}\left[\frac{\bar{\alpha}\lambda e^{\frac{\lambda y}{1+\theta}} }{\left(1+\theta\right)^{2}}+ \frac{2\alpha\lambda e^{\frac{\lambda y}{1+2\theta}} }{\left(1+2\theta\right)^{2}}\right]\\
=\;&\theta \left[\frac{\bar{\alpha}\lambda e^{\frac{-\theta\lambda y}{1+\theta}} }{\left(1+\theta\right)^{2}}+ \frac{2\alpha\lambda e^{\frac{-2\theta\lambda y}{1+2\theta}} }{\left(1+2\theta\right)^{2}}\right].
\end{split}
\end{equation*}
\end{proof}

\begin{theorem} If we assume a Poisson Transmuted Exponential distribution with parameter $(\alpha,\theta)$ as primary distribution and an Erlang distribution $(2,\lambda)$ with parameter $(\lambda)>0$, as secondary distribution then the pdf of aggregate loss random variable $S= \sum_{i=0}^{X}Y_{i}$ is given by \\
\begin{equation*}
f_{S}(y)=\lambda\theta e^{-\lambda y}\left(\frac{\bar{\alpha}\sinh\left( \frac{y\lambda}{\sqrt{(1+\theta)}}\right)}{(1+\theta)^{3/2}} +\frac{2\alpha \sinh\left(\frac{y\lambda}{\sqrt{(1+2\theta)}} \right) }{(1+2\theta)^{3/2}}\right)\\
\end{equation*}
with \begin{equation*}
f_{S}(0)=\theta\left(\frac{\bar{\alpha}}{(1+\theta)}+\frac{2\alpha}{(1+2\theta)} \right).
\end{equation*}
\end{theorem}
\begin{proof} By assuming that the claim severity follows Erlang (2,$\lambda$) distribution, then the $n ^{th}$ fold convolution of Erlang distribution is gamma distribution with parameter $(2n, \lambda)$ .The $n^{th}$ fold convolution is given by \\
\begin{equation*}
f^{*x}(y)=\frac{\lambda^{2x}}{(2x-1)!}y^{2x-1}e^{-\lambda y} , x=1,2,...,
\end{equation*}
Then, the pdf of the aggregate random variable S is given by \\
\begin{equation*}
\begin{split}
f_{S}(y)=& \sum\limits_{x=1}^{\infty}\theta\left( \frac{\bar{\alpha}}{(1+\theta)^{x+1}}+\frac{2\alpha}{(1+2\theta)^{x+1}}\right) \frac{\lambda^{2x}y^{2x-1}e^{-\lambda y}}{(2x-1)!}\\
=& \theta e^{-\lambda y}\sum\limits_{x=1}^{\infty}\theta\left( \frac{\bar{\alpha}}{(1+\theta)^{x+1}}\lambda^{2x}\frac{y^{2x-1}}{(2x-1)!}+\frac{2\alpha}{(1+2\theta)^{x+1}}\lambda^{2x}\frac{y^{2x-1}}{(2x-1)!}\right) \\
=& \lambda\theta e^{-\lambda y}\left(\frac{\bar{\alpha}\sinh\left( \frac{y\lambda}{\sqrt{(1+\theta)}}\right)}{(1+\theta)^{3/2}} +\frac{2\alpha \sinh\left(\frac{y\lambda}{\sqrt{(1+2\theta)}} \right) }{(1+2\theta)^{3/2}}\right).
\end{split}
\end{equation*}
\end{proof}

\section{Data analysis}
\subsection{Count data modelling}
In this section, the applicability of Poisson Transmuted Exponential distribution has been shown by considering a data set representing epileptic seizure counts considered earlier by Albert(1991) and Hand et. al.(1994) p. 133 and has compared with following distributions namely\\
\begin{enumerate}
\item[i] Generalized Poisson-Lindley Distribution($\mathcal{GPL}(\alpha,\theta)$)(Mahmoudi and Zakerzadeh(2010)):
\begin{equation*}
f(x;\alpha,\theta)=\frac{\Gamma(x+\alpha)}{x!\Gamma(\alpha+1)}\frac{\theta^{\alpha+1}}{(1+\theta)^{1+x+\alpha}}\left( \alpha+\frac{x+\alpha}{1+\theta}\right).
\end{equation*}
\item[ii] Weighted Generalized Poisson Distribution
    ($\mathcal{WGPD}(a,s,b)$)(Chakraborty (2010)):
\begin{equation*}
G(u)=t^{-s}\frac{K(at,s,bt)}{K(a,s,b)}, \quad \text{where} \quad u=te^{b(1-t)}.
\end{equation*}
\item[iii] Poisson Distribution $(\mathcal{P(\lambda)})$:
\begin{equation*}
P(x)=\frac{e^{-\lambda}\lambda^{x}}{x!},  \lambda>0.
\end{equation*}
\item[iv] Negative Binomial Distribution ($\mathcal{NB(}r,p)$):
\begin{equation*}
f(x)=\binom{x+r-1}{x}p^{x}(1-p)^{r},         r>0,   p\in(0,1).
\end{equation*}
\item[v] A New Generalized Two Parameter Poisson-Lindley Distribution ($\mathcal{NGPL}(\alpha,\theta)$) Bhati et al. (2015)):
\begin{equation*}
f(x;\alpha,\theta)=\frac{\theta^{2}}{(\theta+\alpha)(1+\theta)^{x+1}}\left( 1+\frac{\alpha(x+1)}{(1+\theta)}\right),\quad x=0,1,\dots
\end{equation*}
\end{enumerate}

\begin{table}[htbp]
  \centering
  \caption{Distribution of epileptic seizure counts}
\small\addtolength{\tabcolsep}{-3pt}
    \begin{tabular}{rrrrrrrr}
    \hline
          & Observed & \multicolumn{6}{c}{Expected frequency} \\
    \hline
    Count & Frequency & P & $\mathcal{NB}$  & $\mathcal{WGP}$  & $\mathcal{GPL}$   & $\mathcal{NGPL}$ & $\mathcal{PTE}$ \\
    0     & 126   & 74.94 & 91    & 118.11 & 121.51 & 122   & 121.925 \\
    1     & 80    & 115.71 & 86.6  & 95.81 & 92    & 91    & 91.6166 \\
    2     & 59    & 89.34 & 63.37 & 59.89 & 59    & 58.74 & 58.5609 \\
    3     & 42    & 46    & 42.57 & 34.49 & 35.1  & 35.22 & 34.7734 \\
    4     & 24    & 17.75 & 27.6  & 19.24 & 20.1  & 20.52 & 19.8403 \\
    5     & 8     & 5.48  & 17.6  & 10.59 & 11.18 & 11.22 & 11.0557 \\
    6     & 5     & 1.41  & 10.5  & 5.81  & 6.1   & 6.39  & 6.07055 \\
    7     & 4     & 0.31  & 6.52  & 3.18  & 3.3   & 3.25  & 3.30188 \\
    8     & 3     & 0.06  & 5     & 3.88  & 2.71  & 2.5   & 3.85564 \\ \hline
    Total & 351   & 351   & 351   & 351   & 351   & 351   & 350.99997 \\ \hline
          &       &       &       &       &       &       &  \\
          & parameter & $\hat{\lambda}=1.544$ & $\hat{r}=1.757$ & $\hat{a}=1.089$ & $\hat{\theta}$=1.139 & $\hat{\theta}=1.116$ & $\hat{\alpha}=-0.701$ \\
          &       &       &  $\hat{p}=0.463$ &  $\hat{b}=0.295$ &  $\hat{\alpha}=1.292$ &  $\hat{\alpha}=2.906$ &  $\hat{\theta}=0.873$ \\
          &       &       &       & $\hat{s}$=-1      &       &       &  \\ \hline
          & log likelihood & --636.05 & --595.22 & --595.83 & --594.61 & --594.48 & \textbf{--594.85} \\
          & chi-square & 256.54 & 22.53 & 7.12  & 5.94  & 5.75  & \textbf{5.36} \\
    \hline
    \end{tabular}%
  \label{tab:addlabel}%
\end{table}%

\subsection{Count regression including covariates}
In last two decade researchers have contributed significantly in the area of counts regression modelling, for example Duan et al.(1983), Christensen et al.(1987), Cameron et al.(1998), Cartwright et al.(1992) and Deb and Trivedi(1997), G\'omez--D\'eniz (2010). Here, in this section, we present the application of the proposed distribution in modelling the situation when the count variable $(Y_i)$ also called response variable depends on one or more than one exogenous variable $(\underline{x_i}=x_1,x_2,\cdots,x_s)$. \\
In order to achieve this goal, let us consider the following re-parametrize
\begin{equation}
\nu=2-\alpha \qquad \text{and} \qquad  \theta=\frac{\nu}{2\mu},
\end{equation}
where, $\mu_i$ is the mean of the response variable $(Y_i)$ is related with set of independent variables with log-link function $\log\left(\mu_i \right)=x_i^\top\beta$, for $i=1,2,...,s$ and $\beta$ is a vector of unknown regression coefficient. \\
The above link function will ensure the positivity of $\mu_i$ and after re-parametrization, the log-likelihood of the new model including covariates will be written as
\begin{equation}
l_n(\nu,\underline{\beta}|y_i, \underline{x_i})= \sum\limits_{i=1}^{n} \log \left( \frac{\nu(\nu-1)}{2\mu_i\left(1+\frac{\nu}{2\mu_i}\right)^{y_i+1}}+\frac{(2-\nu)\nu}{\mu_i\left(1+\frac{\nu}{\mu_i} \right)^{y_i+1}} \right)
\end{equation}
and the normal equations are
The parameters $\left(\nu,\beta_1,\beta_2,...,\beta_s\right)$ in the above log-likelihood function can be estimated by maximizing the log-likelihood function for given dataset using optim( ) function in R-program, and the initial values of the parameters were chosen from Poisson regression model.

\subsection{Illustrative example}

The US National Medical Expenditure Survey 1987/88 (NMES) data were considered which can be obtained from Journal of Applied Econometrics,  1997 Data Archive at http://qed.econ.queensu.ca/jae/1997-v12.3/ deb-Trivedi/. This data were originally used by Deb and Trivedi(1997) in their analysis of various measures of healthcare utilization. The data set consists of 4406 individuals covered by Medicare, the USA public insurance programme. Just for illustration, we are considering the Here we model, the number of stays after hospital admission (HOSP) as the response variable, because of the fact that, it have over-dispersed as well as large proportion of zero. The details and the summary statistics of this response variable as well as the set of explanatory variables is given in table (2). Here the mean and variance of the response variable indicate the over-dispersion as well as existence of large number of zeros. Hence it is adequate to apply the our model for the presented dataset.\\

\begin{tiny}
\begin{table}[htbp]
  \centering
  \caption{Summary Statistics of Response and Explanatory Variables}
    \begin{tabular}{llrr} \hline
    \multicolumn{1}{c}{Variable} & \multicolumn{1}{c}{Measurement} & \multicolumn{1}{c}{Mean} & \multicolumn{1}{c}{std dev} \\ \hline
    OFP   & Number of physician visit & 5.774 & 6.759 \\
    HOSP  & Number of hospital stays & 0.296 & 0.746 \\
    POORHLTH & Self-perceived health status, & 0.13  & 0.33 \\
    &poor=1, else=0 & & \\
    EXCLHLTH & Self-perceived health status,  & 0.08  & 0.27 \\
    & excellent =1, else = 0 && \\
    NUMCHRON & Number of chronic conditions & 1.54  & 1.35 \\
    MALE  & Gender; male = 1, else =0 & 0.404 & 0.491 \\
    SCHOOL & Number of year of education & 10.29 & 3.739 \\
    PRIVINS & Private insurance indicator, yes =1, no = 0 & 0.776 & 0.417 \\ \hline
    \end{tabular}
\end{table}
\end{tiny}

In Table (2), the maximum likelihood estimates of  Poisson $(Pois)$ regression Model, Negative Binomial $(\mathcal{NB})$ Regression model and Poisson-Transmuted Exponential$(\mathcal{PTE})$ Regression Model, including the intercept $\beta_1$ (the regression estimate when all variables in the model are evaluated at zero) were presented. For comparability of these models we use the value of the maximum log-likelihood function $(l_{max})$ and Akaike information criterion $(AIC)$ defined as $-2l_{max}+2k$, where $k$ be the no. of parameters in the regression model. These values are shown in Table (3).
\begin{landscape}
\begin{table}[htbp]
  \centering
  \caption{Maximum likelihood estimates their standard error, $t$-values and $p$-values of regression parameters}
    \begin{tabular}{lcccccccccccccc} \hline
    & \multicolumn{4}{c}{$\textit{Pois}$-Regression Model}  &       & \multicolumn{4}{c}{$\mathcal{NB}$-Regression}  &       & \multicolumn{4}{c}{$\mathcal{PTE}$-Regression Model} \\
    \cmidrule(r){2-5}  \cmidrule(r){7-10} \cmidrule(r){12-15}
          &       &       &       &       &       &       &       &       &       &       &       &       &       &  \\
          & estimate & s.e.  & $t$-value & $p$-value &       & estimate & s.e. & $t$-value   & $p$-value &       & estimate & s.e.  & $t$-value     & $p$-value \\ \hline
    Intercept$(\beta_1)$  & 1.0289 & 0.0238 & 43.2580 & 0.0000 &       & 0.9293 & 0.0546 & 17.0220 & 0.0000 &       & --3.5740 & 0.0915 & --39.0754 & 0.0000 \\
              &       &       &       &       &       &       &       &       &       &       &       &       &       &  \\
    HOSP$(\beta_2)$  & 0.1648 & 0.0060 & 27.4780 & 0.0000 &       & 0.2178 & 0.0202 & 10.7930 & 0.0000 &       & --0.0661 & 0.0298 & --2.2142 & 0.0268 \\
              &       &       &       &       &       &       &       &       &       &       &       &       &       &  \\
    POORHLTH$(\beta_3)$ & 0.2483 & 0.0178 & 13.9150 & 0.0000 &       & 0.3050 & 0.0485 & 6.2880 & 0.0000 &       & 1.0481 & 0.0705 & 14.8583 & 0.0000 \\
              &       &       &       &       &       &       &       &       &       &       &       &       &       &  \\
    EXCLHLTH$(\beta_4)$ & --0.3620 & 0.0303 & --11.9450 & 0.0000 &       & --0.3418 & 0.0609 & -5.6100 & 0.0000 &       & --1.7797 & 0.5643 & --3.1535 & 0.0016 \\
              &       &       &       &       &       &       &       &       &       &       &       &       &       &  \\
    NUMCHRON$(\beta_5)$ & 0.1466 & 0.0046 & 32.0200 & 0.0000 &       & 0.1749 & 0.0121 & 14.4660 & 0.0000 &       & 0.0435 & 0.0139 & 3.1303 & 0.0017 \\
              &       &       &       &       &       &       &       &       &       &       &       &       &       &  \\
    MALE$(\beta_6)$  & --0.1123 & 0.0129 & --8.6770 & 0.0000 &       & --0.1265 & 0.0312 & -4.0520 & 0.0000 &       & 0.4583 & 0.0645 & 7.1032 & 0.0000 \\
              &       &       &       &       &       &       &       &       &       &       &       &       &       &  \\
    SCHOOL$(\beta_7)$ & 0.0261 & 0.0018 & 14.1820 & 0.0000 &       & 0.0268 & 0.0044 & 6.1030 & 0.0000 &       & 0.1525 & 0.0112 & 13.5600 & 0.0000 \\
              &       &       &       &       &       &       &       &       &       &       &       &       &       &  \\
    PRIVINS$(\beta_8)$ & 0.2017 & 0.0169 & 11.9630 & 0.0000 &       & 0.2244 & 0.0395 & 5.6860 & 0.0000 &       & --0.1889 & 0.0333 & --5.6730 & 0.0000 \\
              &       &       &       &       &       &       &       &       &       &       &       &       &       &  \\
    dispersion & -     & -     & -     & -     &       & 1.2066 & 0.0336 & 35.5400 & 0.0000 &       & 0.0128 & 0.0002 & 67.0817 & 0.0000 \\ \hline
    \end{tabular}%
  \label{tab:addlabel}%
\end{table}%
\end{landscape}

\begin{table}[htbp]
  \centering
  \caption{log-likelihood $(l_{max})$ and Akaike information criterion $(AIC)$ Value}
    \begin{tabular}{lrrr} \hline
          & \multicolumn{3}{c}{Regression Model} \\  \cmidrule(r){2-4}
    Criterion & \multicolumn{1}{c}{\textit{Pois}} & \multicolumn{1}{c}{$\mathcal{NB}$} & \multicolumn{1}{c}{$\mathcal{PTE}$ } \\ \hline
    $LL_{max}$ & -17971.50 & -12170.55 & -10521.89 \\
    AIC   & 35959 & 24359 & 21061.78 \\ \hline
    \end{tabular}%
\end{table}%

\noindent Further, it is observed that the estimates of all parameters are found significant at $5\%$ level of significance as p-value for all the parameter estimates are less then $5\%$level of significance. It can be clearly seen from table that the log-likelihood value of $\mathcal{PTE}$ regression model is highest, whereas AIC also indicate the better fit of $\mathcal{PTE}$ regression model compared to other two models.

\section{Comments}
In this paper we have proposed a new mixed Poisson distribution and derived its distributional properties. Actuarial application to risk modelling were presented. Moreover, fitting of $\mathcal{PTED}$ indicates the flexibility and capacity of the proposed distribution in data modeling of count data and also in count regression including covariate set up. It has been shown that this distribution out performs other competing distributions.

\end{document}